\documentclass[11pt]{article}
\usepackage{amsfonts,amsmath}
\usepackage[dvips]{graphicx,psfrag}
\usepackage{comment}
\usepackage{ulem} % defines \sout (=striketout}; turn off/on by \normalem/\ULforem
\usepackage{color}
\usepackage{url}
\newcommand{\myurl}[1]{{\small\url{#1}}}

\input{psfig.sty}
\input{epsf.sty}
\includecomment{comment}
\definecolor{two}{rgb}{0,0,0} % highlight of the "2" multiplying \pi
\definecolor{lightblue}{rgb}{.1,.4,0.8}
\specialcomment{hide}{\sf\color{lightblue}}{\rm\color{black}}
\excludecomment{hide}
\specialcomment{lighthide}{\sf\color{magenta}}{\rm\color{black}}
\excludecomment{lighthide}

\definecolor{myred}{rgb}{0,0,0} % black

\newcommand{\mar}[1]{}

\newcommand{\Oh}{\mathrm{O}}

\renewcommand{\sout}[1]{}
\newtheorem{theorem}{Theorem}[section]
\newtheorem{lemma}[theorem]{Lemma}
\newtheorem{corollary}[theorem]{Corollary}

\newcommand{\defi}[1]{\textbf{#1}}
\newcommand{\BbbQ}{\mathbb{Q}}% racionais
\newcommand{\QQQ}{\BbbQ_{\scriptscriptstyle\,\geq}}% racionais>=0

\newcommand{\Acal}{\mathcal{A}}
\newcommand{\Bcal}{\mathcal{B}}

\newcommand{\Hcal}{\mathcal{H}}

\newcommand{\Lcal}{\mathcal{L}}
\newcommand{\Pcal}{\mathcal{P}}

\newcommand{\Mcal}{\mathcal{M}}
\newcommand{\Ncal}{\mathcal{N}}

\newcommand{\barra}{\overline}

\newcommand{\Scal}{\mathcal{S}}
\newcommand{\Ucal}{\mathcal{U}}
\newcommand{\Zcal}{\mathcal{Z}}

\newcommand{\Tstar}{O}
\newcommand{\tstar}{o}

\newcommand{\capt}%
           [2]%
           {\centering\parbox{#1}{\caption[dummy]{\small #2}}}% keep "dummy" as is!

\newenvironment{myitemize}%
   {%
    \begin{list}%
      {}%
      {%
       \setlength{\partopsep}{1.0ex}%
       \setlength{\topsep}{1.0ex}%
       \setlength{\parsep}{0ex}%
       \setlength{\itemsep}{0.5ex}%
       \setlength{\leftmargin}{10,5ex}% \parindent+\labelwidth+\labelsep + 0.5ex
       \setlength{\rightmargin}{0ex}%
       \setlength{\itemindent}{0ex}%
       \setlength{\labelsep}{1ex}%
       \setlength{\labelwidth}{5ex}%
      }%
   }{%
     \end{list}%
    }

\def\qed{~\vrule width.15cm height.2cm depth0cm\smallskip}
\newenvironment{proof}{\noindent{\bf Proof.}}{}
\def\enddiscard{}
\long\def\discard#1\enddiscard{}

\newcommand{\mytextstyle}{}
\def\eps{\varepsilon}
\def\yeps{y^{\eps}}
\newcommand{\opt}{\ensuremath{\mathrm{opt}}}
\newcommand{\pcst}{\textsc{PCST}} 

\newcommand{\jmp}{\textsc{JMP}}

\begin{document}

\title{A Note on Johnson, Minkoff and Phillips' Algorithm \\ 
            for the Prize-Collecting Steiner Tree Problem}

\author{{ 
{\small Paulo Feofiloff
\thanks{Departamento de Ci\^encia da Computa\c c\~ao, Instituto de
  Matem\'atica e Estat\'{\i}stica, Universidade de S\~ao Paulo, Rua do
  Mat\~ao 1010, 05508-090 S\~ao Paulo/SP, Brazil. E-mail:
  \texttt{\{pf,cris,cef,coelho\}@ime.usp.br}. 
  Research supported in part by PRONEX/CNPq 664107/1997-4 (Brazil).}}}
\and{
{\small Cristina G.\ Fernandes \footnotemark[1]
% \thanks{Research supported in part by CNPq Proc.\ 301174/97-0 (Brazil).}
}}
\and{
{\small Carlos E.\ Ferreira \footnotemark[1]
%\thanks{Research supported in part by CNPq Proc.\ 300752/94-6 (Brazil).}
}}
\and{
{\small Jos\'e Coelho de Pina \footnotemark[1]}}
}
\date{revised October 2009~\footnotemark[2]
\addtocounter{footnote}{1}
\footnotetext{This paper was originally published as
\myurl{http://www.ime.usp.br/~cris/publ/jmp-analysis.ps.gz}
in 2006.
The present version makes explicit a stronger statement, 
implicit in the original version: that the addressed 
implementation is a Lagrangean preserving $2$-approximation. 
It also introduces some cosmetic changes in notation and
corrects a technical error in the proof of one of the invariants.}
% \\ \textcolor{red}{(mytinkering/jmp-analysis.tex) \today}
%% used to be \textcolor{red}{(mytinkering/techreport2.tex)
}

\maketitle

\begin{abstract}
\noindent
The primal-dual scheme has been used to provide approximation
algorithms for many problems. 
Goemans and Williamson gave a $(2-\frac{1}{n-1})$-approximation 
for the Prize-Collecting Steiner Tree Problem 
that runs in $\Oh(n^3 \log n)$ time.
% it applies the
% primal-dual scheme once for each of the $n$ vertices of the graph.
Johnson, Minkoff and Phillips proposed a faster implementation of
Goemans and Williamson's algorithm. 
We give a proof that the
approximation ratio of this implementation is exactly~$2$.
%\addtocounter{footnote}{-1}
% \textcolor{myred}{(in the worst case)}.
% \vspace*{1.5ex}
% \noindent
% This paper was published as
% \myurl{http://www.ime.usp.br/~cris/publ/jmp-analysis.ps.gz}
% in 2006.
% The present version makes explicit a stronger statement, 
% implicit in the original version: that the addressed 
% implementation is a Lagrangean preserving $2$-approximation. 
% It also introduces some cosmetic changes in notation and
% corrects a technical error in the proof of one of the invariants. 
\end{abstract}

\section{Introduction}

Consider a graph $G=(V,E)$, a function $c$ from $E$ into 
the set $\QQQ$ of non-negative rationals 
and a function $\pi$ from~$V$ into~$\QQQ$.
The \defi{Prize-Collecting Steiner Tree Problem} (\pcst) 
% consists of the following: given $G$, $c$, and~$\pi$, find 
asks for a tree $T$ in $G$ such that
$\sum_{e \in E_T}c_e + \sum_{v \in V \setminus V_T}\pi_v$ is minimum.
(We denote by $V_T$ and~$E_T$, respectively, 
the vertex and edge sets of a graph~$T$.) 
\begin{hide}
  ((We are assuming that, by definition, every tree has at least one vertex;
    this is relevant in case $c=\infty$ and $\pi=0$, for example.
    Come to think of it,
    I am not sure this is really important.
    Maybe I need it in the proof of Lemma~\ref{lemma3}?))
\end{hide}
The rooted variant of the problem requires $T$ to contain a given root
vertex.

Goemans and Williamson~\cite{GoemansW95,Hochbaum97} used a primal-dual
scheme to derive a $(2-\frac{1}{n-1})$-approximation 
for the rooted variant of~\pcst,
where $n:=|V|$. 
By trying all possible choices for the root, they
obtained a $(2-\frac{1}{n-1})$-approximation for the unrooted~\pcst.
The resulting algorithm runs in time $\Oh(n^3 \log n)$.  
Johnson, Minkoff and Phillips~\cite{JohnsonMP00} 
proposed a modification of the
algorithm that runs the primal-dual scheme only once,
resulting in a running-time of $\Oh(n^2 \log n)$. 
They claimed their algorithm~--- which we refer to as \jmp~---
achieves an approximation ratio of~$2-\frac{1}{n-1}$.
Unfortunately, their claim does not hold.
% , as we show below. 

This note does two things.
First, it proves that the \jmp\ algorithm is a $2$-approximation
(the proof involves some non-trivial technical details). 
Second, 
it shows an example where the approximation ratio achieved by the \jmp\ algorithm 
is exactly~$2$, 
thereby contradicting 
the claim by Johnson, Minkoff and Phillips. % \cite{JohnsonMP00}.

\section{Notation and preliminaries}

% The description of the algorithm in the next section will use a
% notation slightly different than the one in the seminal paper by
% Goemans and Williamson~\cite{GoemansW95}. With this notation, we found
% it easier to be sure of the correctness of all the technical details
% in the analysis.

For any subset $F$ of~$E$, let $c(F) := \sum_{e \in F}c_e$.
For any subset $X$ of~$V$, let $\pi(X) := \sum_{v \in X}\pi_v$ 
and let $\barra{X}:=V \setminus X$. 
If $T$ is a subgraph of~$G$, 
we shall abuse notation and write $\pi(T)$ and $\pi(\barra{T})$
to mean $\pi(V_T)$ and $\pi(\barra{V_T})$ respectively.
Similarly, we shall write $c(T)$ 
to mean~$c(E_T)$.
Hence, the goal of $\pcst(G,c,\pi)$ is to find a tree $T$ in $G$ 
such that $c(T)+\pi(\barra{T})$ is minimum.

A collection $\Lcal$ of 
\textcolor{myred}{nonnull}
subsets of $V$ is \defi{laminar} if, 
for any two elements $L_1$ and $L_2$ of~$\Lcal$, 
either $L_1 \cap L_2 = \emptyset$ 
or $L_1 \subseteq L_2$ 
or $L_1 \supseteq L_2$. 
\begin{hide}
 ((We assume that no laminar collection contains the null set.))
\end{hide}
For any subset $X$ of~$V$, let 
\[
\Lcal[X] := \{L\in \Lcal : L\subseteq X\}
\quad\mbox{and}\quad
\Lcal_X  := \{L\in \Lcal : L\supseteq X\}\:.
\]
\begin{hide}
 ((If $X\in \Lcal$ then $X$ is in $\Lcal[X]$ as well as in $\Lcal_X$.))
\end{hide}
For every $L$ in $\Lcal$ that is not in
$\Lcal[X] \cup \Lcal[\barra{X}] \cup \Lcal_X$,
the sets $L\cap X$, $L\setminus X$ and $X\setminus L$ are all nonempty.
For any subgraph $T$ of~$G$,
we shall abuse notation and write 
$\Lcal[T]$, $\Lcal[\barra{T}]$, and $\Lcal_{T}$
in place of 
$\Lcal[V_T]$, $\Lcal[\barra{V_T}]$, and $\Lcal_{V_T}$
respectively.

The union of all sets in~$\Lcal$ shall be denoted by~$\bigcup\Lcal$.
The set of all maximal elements of~$\Lcal$ shall be denoted by~$\Lcal^*$.
If $\Lcal$ is laminar,
the elements of $\Lcal^*$ are pairwise disjoint.
If, in addition, $\bigcup \Lcal = V$ then  $\Lcal^*$ is a partition of~$V$.

For any laminar collection $\Lcal$ of subsets of $V$ 
and any edge $e$ of~$G$, 
let $\Lcal(e) := \{L \in \Lcal: e \in \delta_G L\}$, 
where $\delta_G L$ stands for the set of edges of $G$ 
with one end in $L$ and the other in~$\barra{L}$. 

Let $y$ be a function from $\Lcal$ into~$\QQQ$.
For any subcollection $\Lcal'$ of~$\Lcal$, 
let
$
y(\Lcal'):=\sum^{}_{L \in \Lcal'}y^{}_L
$.
We say that $y$ \defi{respects} $c$~if 
\begin{equation}
y(\Lcal(e)) \ \leq \ c_e 
\quad \mbox{for each $e$ in $E$\:.} \label{yrespeitac}
\end{equation} 
% This inequality is the usual \pcst\ restriction on edges: 
% the sum of $y^{}_L$ for all $L$ in $\Lcal$ that ``cross'' edge $e$
% should not exceed~$c_e$. 
We say an edge $e$ is \defi{tight for~$y$} if equality holds in~(\ref{yrespeitac}). 
We say $y$ \defi{respects} $\pi$~if
\begin{equation}
\mytextstyle
y(\Lcal[X]) \ \leq \ \pi(X) 
\quad  \mbox{for each $X$ in $\Lcal$\:.} \label{yrespeitapi1}
\end{equation} 
% This inequality is the usual \pcst\ restriction
% according to which the sum of~$y^{}_L$ for all $L$ in
% $\Lcal$ contained in~$X$ should not exceed the sum of the ``penalties''
% ($\pi$-values) of all vertices in~$X$.
\begin{hide}
 ((I could have required this for all $X\subseteq V$, 
 but I guess this would be redundant.)) 
\end{hide}
We shall say that 
$y$ \defi{saturates} an element $X$ of~$\Lcal$
if equality holds in~(\ref{yrespeitapi1}). 
The following lemma summarizes the effect of the two ``respects''
constraints on~$y$:

\begin{hide}
 ((The following lemma takes the place of weak duality.
   The $y(\Lcal \setminus \Lcal_{T})$ instead of $y(\Lcal)$
   is an echo of the root vertex in the linear program.))
\end{hide}

\begin{lemma} 
Let $\Lcal$ be a laminar collection of subsets of $V$ 
and $y$ a function from $\Lcal$ into~$\QQQ$.
If $y$ respects $c$ and~$\pi$ then
\[
y(\Lcal \setminus \Lcal_{T}) 
% y(\Lcal) - \sum^{}_{L \supseteq V_T} y^{}_L 
\ \leq \ 
c(T) + \pi(\barra{T})
\]
for any connected subgraph $T$ of~$G$.
\end{lemma}

\begin{proof}
For $\Mcal := \{L \in \Lcal: \delta_T L \neq \emptyset\}$,
we have
$\mytextstyle
y(\Mcal) 
% \ = \ \sum^{}_{L \in \Mcal}y^{}_L 
  \ \leq \ \sum^{}_{L \in \Mcal}|\delta_T L|y^{}_L  
% \ = \ \sum^{}_{e \in E_T} \sum^{}_{L \in \Lcal(e)} y^{}_L 
  \ = \ \sum^{}_{e \in E_T} y(\Lcal(e)) 
  \ \leq \ \sum^{}_{e \in E_T} c_e 
  \ = \ c(T)
$.
For $\Ncal := \Lcal[\barra{T}]$, % \{L \in \Lcal : L \subseteq \barra{V_T}\}
we have
$\mytextstyle
 y(\Ncal) 
 \ =    \ \sum^{}_{L\in \Ncal^*} y(\Lcal[L])
 \ \leq \ \sum^{}_{L\in \Ncal^*} \pi(L)
 \ \leq \ \pi(\barra{T})
$.
The lemma follows from the two inequalities 
since $\Lcal  = \Mcal \cup \Ncal \cup \Lcal_{T}$.~\qed
\end{proof}

\begin{hide}
 ((This lemma will not be used with the tree produced by the algorithm
   playing the role of $T$.
   Instead, it will be used with an optimum tree playing the role of $T$.))
\end{hide}

Let $\opt(\pcst(G,c,\pi))$ denote the minimum value of the sum 
$c(T)+\pi(\barra{T})$ when $T$ is a tree in~$G$.
Then the following corollary establishes the relevant lower
bound for~$\opt(\pcst(G,c,\pi))$:

\begin{corollary}\label{n-coro}
Let $\Lcal$ be a laminar collection of subsets of $V$ 
and $y$ a function from $\Lcal$ into~$\QQQ$.
If $y$ respects $c$ and~$\pi$ then
$y(\Lcal \setminus \Lcal_{\Tstar}) 
% \sum^{}_{L \supseteq V_{\Tstar}} y^{}_L 
\ \leq \ 
\opt(\pcst(G,c,\pi))$
for any optimal solution $\Tstar$ of $\pcst(G,c,\pi)$.~\qed
\end{corollary}

Before we state the algorithm,
a few more definitions are needed.
Let $\Lcal$ be a laminar collection of subsets of $V$ 
such that $\bigcup \Lcal = V$.
We say that an edge is \defi{internal to}~$\Lcal^*$ 
if both of its ends are in the same element of~$\Lcal^*$. 
All other edges are \defi{external to}~$\Lcal^*$. 
For any external edge, there are two
elements of $\Lcal^*$ containing its ends. 
We call these two elements the \defi{extremes} of the edge in~$\Lcal^*$.

Given a forest $F$ in $G$ and a subset $L$ of~$V$,
we say that $F$ is \defi{$L$-connected}
if \textcolor{myred}{$V_F\cap L=\emptyset$ or}
the induced subgraph $F[V_F\cap L]$ is connected.
In other words, 
$F$ is $L$-connected if the following property holds:
for any two vertices $x$ and $y$ of $F$ in~$L$,
there exists a path from $x$ to $y$ in $F$ 
and that path never leaves~$L$.
If $F$ spans~$G$
(as is the case during the first phase of the algorithm below),
the condition ``$F[V_F\cap L]$ is connected'' 
can, of course, be replaced by ``$F[L]$ is connected''.

For any collection $\Lcal$ of subsets of~$V$,
we shall say that $F$ is \defi{$\Lcal$-connected}
if $F$ is $L$-connected for each $L$ in~$\Lcal$.

For any collection $\Scal$ of subsets of~$V$, 
we say a tree $T$ % in $G$ 
\defi{has no bridge in $\Scal$} 
if $|\delta_T S| \neq 1$ 
(whence $\delta_T S = \emptyset$ or $|\delta_T S| \geq 2$) 
for all $S$ in~$\Scal$.
We say that a tree $T$ in $G$ \defi{is wrapped in} $\Scal$
if $V_T \subseteq S$ for some $S$ in $\Scal$.
\begin{hide}
 ((i.e., $\Scal_{T} = \emptyset$.))
\end{hide}

\section{Johnson, Minkoff and Phillips' algorithm}

The \jmp\ algorithm is a $2$-approximation for the~\pcst. 
It receives $G$, $c$,~$\pi$ and returns a tree $T$ in~$G$ 
such that $c(T)+\textcolor{two}{2}\,\pi(\barra{T}) \leq 2\,\opt(\pcst(G,c,\pi))$.
\textcolor{myred}{
  For our purposes, it would be enough to have
  $c(T)+\pi(\barra{T})$ on the left side of the inequality.
  The factor \textcolor{two}{$2$} multiplying $\pi$ is a bonus,
  and, because of it, the \jmp\ algorithm is said to be 
  a {\it Lagrangean preserving} 2-approximation~\cite{ArcherBHK09}.%
}

The algorithm has two phases, 
the second one operating on the output of the first.
% Both phases are iterative.

\medskip
\textbf{Phase I:} 
Each iteration in phase~I starts with 
a spanning forest $F$ in~$G$, 
a laminar collection $\Lcal$ of subsets of $V$ such that $\bigcup\Lcal=V$, 
a subcollection $\Scal$ of~$\Lcal$, 
and a function $y$ from $\Lcal$ into $\QQQ$ 
such that the following invariants hold:
\begin{hide}
 ((Invariant (i0) used to be this:
   ``all edges of $F$ are internal to~$\Lcal^*$.''
   But this is used only to prove that $F$ is a forest, 
   which is a anonymous invariant.
   Hence, (i0) must be treated as a anonymous invariant as well.))
\end{hide}
\begin{myitemize}
\item[(i1)] 
$F$ is $\Lcal$-connected;
\item[(i2)] 
$y$ respects $c$ and~$\pi$;
\item[(i3)] 
each edge of $F$ is tight for~$y$;
\item[(i4)] 
$y$ saturates every element of~$\Scal$;

\item[(i5)] 
\sout{$\bigcup \Scal \neq V$;}
\textcolor{myred}{no element of $\Lcal^*\setminus\Scal$
 is the union of elements of $\Scal$;}
\begin{hide}
  ((OLD: no element of $\Lcal\cup\{V\}$ is the union of 
    two or more elements of~$\Scal$;))
\end{hide}

\item[(i6)] 
for any $\Lcal$-connected tree $T$ in~$G$, 
if $T$ has no bridge in~$\Scal$
and is not wrapped in~$\Scal$
then
\begin{equation}\label{n-eq3}
\mytextstyle
\sum^{}_{e\in E_T} y(\Lcal(e)) 
%+ 2 \sum^{}_{L \subseteq \barra{V_T}} y^{}_L  
 + \textcolor{two}{2}\,y(\Lcal[\barra{T}])  
\ \leq \ 
2\,y(\Lcal \setminus \Lcal_{\{\tstar\}})
\end{equation}
for any vertex $\tstar$ of~$G$. 
\begin{lighthide}
 (This says essentially that,
 for any subtree tree $T$ of~$F$, 
 if $T$ has no bridge in~$\Scal$
 and is not wrapped in~$\Scal$
 then
 $c(T) + 2\,y(\Lcal[\barra{T}]) \leq  
 2\,y(\Lcal \setminus \Lcal_{\{\tstar\}})$
 for any vertex $\tstar$ in any optimum solution of the \pcst.
 For technical reasons, however, the invariant is formulated
 in more general terms.)
\end{lighthide}
\begin{hide}
 ((You cannot write this as
  $c(T) + 2\,y(\Lcal[\barra{T}]) \leq  
  2\,y(\Lcal \setminus \Lcal_{\{\tstar\}})$,
  because then you would have to add hypothesis 
  ``all edges of $T$ are tight for~$y$''
  and that would make it hard/impossible to prove~(i6)
  because you could not apply induction hypothesis to $T$
  whose edges are tight for $\yeps$.))
\end{hide}
\end{myitemize}
\begin{hide}
  ((It would be nice to give some intuition about the meaning of~(i6).
   But I don't know how to do this.))
% (Informally, the inequality may be understood as follows:
% there are ``on average''
% less than $2$ edges of $T$ in $\delta L$ for each $L$ em~$\Lcal$.)
\end{hide}
\begin{hide}
  ((For some misterious reason, we don't seem to need the following invariant:
   for each $L$ in $\Lcal$,
   if $|L|>1$ then $L=L'\cup L''$ for 
   some $L'$ and $L''$ in $\Lcal\setminus\{L\}$.))
\end{hide}
The first iteration starts with 
$F=(V,\emptyset)$, 
$\Lcal = \{\{v\}: v \in V\}$, $\Scal=\emptyset$, 
and $y=0$. 
Each iteration consists of the following:

\begin{hide}
  ((I guess you can eliminate variable $\Scal$ and,
    at the beginning of each iteration, define $\Scal$ as the set 
    of all saturated elements of $\Lcal$.
    Subcase I.1.A would be the case where at least one element of
    $\Lcal^*\setminus\Scal$ becomes saturates. 
    Do nothing in this case.
    Ooops. But then you may miss Case I.2 because they all become saturated
    at once... ))
\end{hide}

\noindent 
\begin{quote}
\textbf{Case I.1:} \
$|\Lcal^* \setminus \Scal| > 1$.\\[0.4ex] 
For $\eps$ in $\QQQ$, let $\yeps$ be the function defined as follows:
$\yeps_L  = y^{}_L + \eps$ if $L \in \Lcal^* \setminus \Scal$
and $\yeps_L = y^{}_L$ otherwise.
Let $\eps$ be the largest number in $\QQQ$ such that the function $\yeps$
respects $c$ and~$\pi$.
\begin{hide}
 ((Hence,
  either (1)~$\yeps(\Lcal[S])=\pi(L)$
  for some $L$ in $\Lcal^*\setminus\Scal$ 
  or (2)~$\yeps(\Lcal(e))=c_e$
  for some external edge $e$ at least one of whose extremes in $\Lcal^*$
  is in $\Lcal^*\setminus\Scal$.
  In both cases we may have $\eps=0$.))
\end{hide}

\begin{quote}
\pagebreak[3]
\textbf{Subcase I.1.A:} \
$\yeps$ saturates some element $L$ of $\Lcal^* \setminus \Scal$.\\[0.4ex]
Start a new iteration with $\Scal \cup \{L\}$ and $\yeps$
in the roles of $\Scal$ and $y$ respectively.
(The forest $F$ and the collection $\Lcal$ do not change.)

\smallskip
\textbf{Subcase I.1.B:} \ 
% \sout{Subcase~I.1.A does not hold but}
some edge $e$ external to $\Lcal^*$ 
is tight for~$\yeps$
\textcolor{myred}{
and has at least one of its extremes in $\Lcal^*\setminus\Scal$.}\\[0.4ex]
Let $L_1$ and $L_2$ be the extremes of~$e$ in~$\Lcal^*$.
Set $\yeps_{L_1 \cup L_2} := 0$ and start a new iteration with 
$F+e$, $\Lcal \cup \{L_1 \cup L_2\}$, and $\yeps$ 
in the roles of $F$, $\Lcal$, and $y$ respectively.
\begin{hide}
  (($F+e$ is still a forest because $F$ had no edges external to $\Lcal^*$.
    This  invariant is anonymous, just as ``$F$ is forest'' is anonymous.))
\end{hide}
(The collection $\Scal$ does not change.)
\end{quote}

\begin{hide}
 ((Example where the order in which subcases I.1.A and I.1.B are executed
   makes a difference:
   $G$ is circuit with vertices $1,\ldots,n$,
   $\pi_1=\pi_2=10$ and all other $\pi$ are $1$,
   $c_{12}=2+\rho$ and all ther $c$ are $2$.
   (But the algorithm is correct no matter the order in which subcases are
   executed.))
\end{hide}

%%%%%%%%%%%%% delete very soon
% \begin{hide}
% \bf
%  ((Humm...
%    In Subcase~I.1.B, it is possible that $L_1$ and $L_2$ are in $\Scal$
%    (and $\yeps=0$).
%    This smells fishy...
%    A pair of inactive sets may come alive?))
% \end{hide}

\textbf{Case I.2:} \
$|\Lcal^* \setminus \Scal| = 1$. \\[0.4ex]
\begin{hide}
  ((Nihil obstat $\Lcal^*\setminus\Scal = \{V\}$
    and nihil obstat the only element of $\Lcal^*\setminus\Scal$ being saturated.))
\end{hide}
This is the end of phase~I.
Start phase~II.
\end{quote}

\noindent
\textbf{Phase II:} 
During this phase, the collections $\Lcal$ and $\Scal$
and the function $y$ remain unchanged.
Let $M$ be the only element of $\Lcal^* \setminus \Scal$.
Each iteration begins with a % \sout{subtree} 
subgraph~$T$
\begin{hide}
  ((with at least one vertex))
\end{hide}
of $F$ such that 
\begin{myitemize}

\item[(i7)] 
$T$ is an $\Lcal$-connected tree;

% \item[(i8)] 
% $T$ is not wrapped in~$\Scal$;

\item[(i8)] 
\sout{$\barra{V_T}$ admits a partition into elements of $\Scal$.}
\textcolor{myred}{$M\setminus V_T$ admits a partition into elements of $\Scal$.}
\begin{hide}
 ((equivalent: $M\setminus V_T$ is the union of some elements of $\Scal$))
\end{hide}

% \item[(i10)] $\pi(\barra{T}) = y(\Lcal[\barra{T}])$;
%%%  I don't say ``$y$ saturates $\barra{V_T}$'' 
%%%  because $V_T$ may not be in $\Lcal$.

% \item[(i11)] for each $S$ in $\Scal$, 
%              if $|\delta_T S|=1$ then $S\subseteq V_T$.
\end{myitemize}
The first iteration begins with $T=F[M]$.
Each iteration does the following:
\begin{hide}
  ((At the beginning of phase~II,  
   the expression ``$F[M]$'' makes sense because $M\subseteq V_F$.
   Hence, $T$ is well defined.))
\end{hide}

\smallskip
\noindent 
\begin{quote}
\textbf{Case II.1:} \
$|\delta_T Z| = 1$ for some $Z$ in~$\Scal$. \\[0.4ex]
Start a new iteration with $T-Z$ in place of~$T$.
\begin{hide}
  ((Choosing a maximal $Z$ is, I guess, technically correct,
    but creates a headache\ldots \
    You must show that the following situation cannot occur:
    $S'\subset S''$, both saturated, $|\delta_T S'|=1$, $|\delta_T S''|=2$
    but $|\delta_{T-S'} S''|=1$.))
\end{hide}

\textbf{Case II.2:} \
$|\delta_T Z| \neq 1$ for each $Z$ in~$\Scal$.\\[0.4ex]
Return $T$ and stop.
\end{quote}

\section{Analysis of the algorithm}

Suppose, for the moment, that invariants~(i1) to~(i8) are correct.
At the end of phase~II,
$T$ is a tree by virtue of~(i7)\mar{(i7)}.
As $T$ is a subgraph of $F$, due to~(i3)\mar{(i3)},
\[\mytextstyle
c(T) 
= \sum_{e \in E_T} c_e 
= \sum_{e \in E_T} y(\Lcal(e))\:.
%     \ = \ \sum_{L \in \Lcal} |\delta_T L|y^{}_L\:.
\]
On the other hand,
\textcolor{myred}{
$\Lcal^*\cap\Scal$ is a partition of $\barra{M}$
and,
by~(i8)\mar{(i8)},
there is a partition of $M\setminus V_T$ into elements of~$\Scal$.
Therefore, some subcollection $\Zcal$ of $\Scal$ 
is a partition of~$\barra{V_T}$.
}
Hence,
\[
\mytextstyle
\pi(\barra{T})
 =  \sum^{}_{S\in \Zcal} \pi(S)
 =  \sum^{}_{S\in \Zcal} y(\Lcal[S])
\leq y(\Lcal[\barra{T}])\:.
\]
Here, the second equality follows from~(i4)\mar{(i4)}.
\begin{hide}
  ((Cannot justify equality by ``$y$ respects $\pi$''
   because $\barra{V_T}$ is not necessarily in $\Lcal$.
   But equality will be forced below anyway.)
\end{hide}
Therefore,
\begin{equation}\label{n-zero}
\mytextstyle
c(T) + \textcolor{two}{2}\pi(\barra{T})  
\leq 
\sum_{e \in E_T} y(\Lcal(e)) + \textcolor{two}{2}\,y(\Lcal[\barra{T}])\:. 
\end{equation}
In order to show that~(\ref{n-eq3}) holds,
we must verify that $T$ satisfies the hypotheses of~(i6).
By~(i7)\mar{(i7)},
$T$ is $\Lcal$-connected.
\textcolor{myred}{
Due to~(i5)\mar{(i5)}, $M$ is not the union of elements of~$\Scal$.
Hence,  
by virtue~(i8)\mar{(i8)},
$T$ is not wrapped in~$\Scal$.
}
Since we are in Case~II.2,
$T$ has no bridge in~$\Scal$.
Hence, $T$ satisfies the hypotheses of~(i6)\mar{(i6)}.
Now, by~(\ref{n-eq3}) coupled with~(\ref{n-zero}),
\begin{equation}\label{ph2:one}
\mytextstyle
c(T) + \textcolor{two}{2}\,\pi(\barra{T}) 
\leq  2\,y(\Lcal \setminus \Lcal_{\{\tstar\}})
\end{equation}
for any vertex~$\tstar$.
Now, let $\tstar$ be an arbitrary vertex 
of an optimal solution $\Tstar$ of $\pcst(G,c,\pi)$.
Since $y$ respects $c$ and~$\pi$,
as stated in~(i2),\mar{(i2)}
Corollary~\ref{n-coro} implies
\[
c(T) + \textcolor{two}{2}\,\pi(\barra{T})
\ \leq \ 2\,y(\Lcal \setminus \Lcal_{\{\tstar\}})
\ \leq \ 2\,y(\Lcal \setminus \Lcal_{\Tstar})
\ \leq \ 2\,\opt(\pcst(G,c,\pi))\:.
\]
This proves the following theorem
(which is the correct version of Theorem~3.2 
by Johnson, Minkoff and Phillips~\cite{JohnsonMP00}):

\begin{theorem} \label{n-maint}
The \jmp\ algorithm is a Lagrangean preserving $2$-approximation for the~\pcst.
\end{theorem}

To complete the proof of the theorem 
we must only verify the invariants of the algorithm, 
something we shall do in the next section. 

The example in Figure~\ref{example} shows that the approximation ratio
of the \jmp\ algorithm can be arbitrarily close to~$2$, 
regardless of the size of the graph.
% (Therefore Theorem~3.2 in~\cite{JohnsonMP00} does not hold.) 
So, Theorem~\ref{n-maint} is tight.

\begin{figure}[tbh]
\begin{center}
\psfrag{a}{{\small $1+\rho$}}
\includegraphics{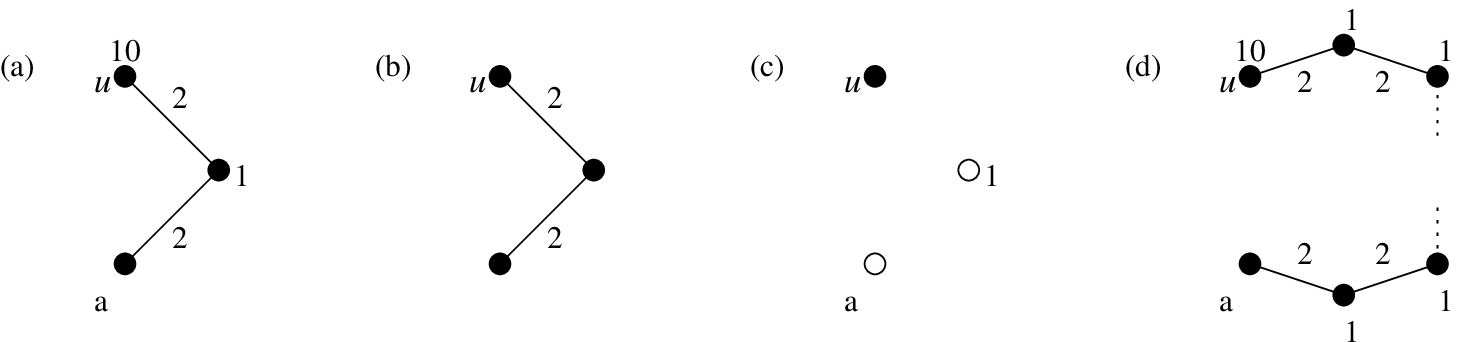} 
\capt{91ex}{\label{example}
(a)~An instance of the~\pcst.
(b)~The solution produced by the \jmp\ algorithm when $\rho>0$. Its cost is 4. 
(c)~The optimal solution, consisting of vertex~$u$ alone, has cost~$2+\rho$.
(d)~A similar instance of arbitrary size consists of a long path.
% To avoid wrong interpretation,
% it would be better if the dotted lines would start solid 
% and only then become dotted.
% Better yet: give a concrete example with 6 vertices.
}
\end{center}
\end{figure}

\section{Proofs of the invariants}

Invariants (i1) to (i4) obviously hold
at the beginning of each iteration of phase~I.
We must only verify the other four invariants.

\color{myred}
\subparagraph{Proof of (i5).}
Obviously (i5) holds at the beginning of the first iteration.
Now consider an iteration where Case~I.1 occurs.
If Subcase~I.1.A occurs, 
then (i5) remains trivially true at the beginning of the next iteration.
Next, suppose Subcase~I.1.B occurs.
Adjust notation so that~$L_1\notin\Scal$.
Since (i5)\mar{(i5)} holds at the beginning of the current iteration,
$L_1$ is not the union of elements of~$\Scal$. 
Hence, $L_1\cup L_2$ is not the union of elements of~$\Scal$.
Therefore,
(i5) remains trivially true at the beginning of the next iteration.~\qed
\color{black} 

The verification of~(i6)
depends on the following lemma:

\begin{lemma}
\label{lemma3}
Let $\Pcal$ be a partition of $V$
and $(\Acal,\Bcal)$ a bipartition of~$\Pcal$.
Let $T$ be a tree in~$G$.
If $T$ is $\Pcal$-connected,
has no bridge in~$\Bcal$,
and is not wrapped in~$\Bcal$,
then
\begin{equation}\label{n-equality-lemma}
 \mytextstyle
\frac{1}{2} \sum^{}_{A \in \Acal}|\delta_T A| 
           + |\Acal[\barra{T}]|
 \ \leq \ |\Acal| - 1\:.
\end{equation}
\end{lemma}

\begin{proof}
Let us say that two elements of $\Pcal$
are \textit{adjacent} if there is an edge of $T$ 
with these two elements as extremes.
This adjacency relation defines a graph $\Hcal$
having $\Pcal$ as set of vertices.
Since $T$ is $\Pcal$-connected,
the edges of $\Hcal$ are in one-to-one correspondence
with the edges of $T$ external to~$\Pcal$.
% for any given pair $(A_1,A_2)$ of elements of $\Acal$,
% there is at most one edge of $T$ with extremes $A_1$ and~$A_2$,
Hence,
the degree of any vertex $P$ of $\Hcal$ is exactly~$|\delta_T P|$,
and therefore
$\frac{1}{2} \sum_{P \in \Pcal}|\delta_T P| = |E_{\Hcal}|$.
% Let $h$ denote the number of components of~$\Hcal$.
Since $T$ is connected,
$\Hcal$ has $1 + |\Pcal[\barra{T}]|$ components
(all are singletons, except at most one).
\begin{hide}
 ((I guess here we need $V_T\neq \emptyset$.))
\end{hide}
Since $T$ has no cycles and is $\Pcal$-connected,
$\Hcal$ is a forest.
Hence
$|E_{\Hcal}| = |\Pcal| - 1 - |\Pcal[\barra{T}]|$
and therefore
\begin{equation}\label{eq1:lemma3}
\mytextstyle
\frac{1}{2} \sum^{}_{P\in\Pcal}|\delta_T P|
 =      |\Pcal|-1-|\Pcal[\barra{T}]|\:.
\end{equation}
Now consider the vertices of $\Hcal$ that are in~$\Bcal$.
Since $T$ has no bridge in~$\Bcal$ and is not wrapped in~$\Bcal$, 
each $B$ in $\Bcal$ is such that either
$|\delta_T B| \geq 2$ or $B\subseteq \barra{V_T}$.
Hence $\sum^{}_{B \in \Bcal} |\delta_T B| 
\geq 2\,|\Bcal \setminus \Bcal[\barra{T}]|$,
and therefore
\begin{equation}\label{eq2:lemma3}
\mytextstyle
\frac{1}{2} \sum^{}_{B \in \Bcal} |\delta_T B| 
\ \geq \ |\Bcal| - |\Bcal[\barra{T}]|\:.
\end{equation}
The difference between~(\ref{eq1:lemma3}) and~(\ref{eq2:lemma3})
is the claimed inequality~(\ref{n-equality-lemma}).~\qed
\end{proof}

\subparagraph{Proof of (i6).}
It is clear that (i6) holds at the beginning of the first iteration.
Now assume that it holds at
the beginning of some iteration where Case~I.1 occurs. 

Suppose, first, that Subcase~I.1.A occurs.
At the end of the subcase,
let $\Scal':=\Scal\cup\{L\}$,
let $\tstar$ be any vertex, 
and let $T$ be an $\Lcal$-connected tree
that has no bridge in~$\Scal'$, 
is not wrapped in~$\Scal'$,
\textcolor{myred}{and such that all its edges are tight for~$\yeps$}.
\textcolor{myred}{Of course all edges of $T$ are tight for~$y$}.
Since $T$ has no bridge in~$\Scal$
and is not wrapped in~$\Scal$, (\ref{n-eq3}) holds.
We must show that (\ref{n-eq3}) also holds when $\yeps$ is substituted for~$y$.
Let $\Pcal:=\Lcal^*$, 
$\Acal:=\Lcal^*\setminus\Scal$, 
and $\Bcal:=\Lcal^*\cap\Scal$.
Since $|\Acal_{\{\tstar\}}| \leq 1$,
Lemma~\ref{lemma3} implies
\[
\mytextstyle
\sum^{}_{A \in \Acal}|\delta_T A|\,\eps 
           + \textcolor{two}{2}\,|\Acal[\barra{T}]|\,\eps
 \ \leq \ 2\,|\Acal\setminus\Acal_{\{\tstar\}}|\,\eps\:.
\]
The addition of this inequality to~(\ref{n-eq3})
% (left side to the left side and right side to the right side)
produces 
\[
\mytextstyle
\sum^{}_{e\in E_T} \yeps(\Lcal(e)) 
 + \textcolor{two}{2}\,\yeps(\Lcal[\barra{T}])  
\ \leq \ 
2\,\yeps(\Lcal \setminus \Lcal_{\{\tstar\}})\:,
\]
since $\yeps$ differs from $y$ only in~$\Acal$.
Hence, (i6) remains true at the beginning of the next iteration.

Now suppose Subcase~I.1.B occurs.
At the end of the subcase,
let $\Lcal':=\Lcal\cup\{L_1 \cup L_2\}$,
let $\tstar$ be any vertex,
and let $T$ be an $\Lcal'$-connected tree
that has no bridge in~$\Scal$
and is not wrapped in~$\Scal$.
Since $T$ is $\Lcal$-connected, (\ref{n-eq3}) holds.
We must show that (\ref{n-eq3}) remains true when $\yeps$ and $\Lcal'$
are substituted for $y$ and $\Lcal$ respectively.
Let $\Pcal:=\Lcal^*$, 
$\Acal:=\Lcal^*\setminus\Scal$, 
and $\Bcal:=\Lcal^*\cap\Scal$.
Since $|\Acal_{\{\tstar\}}| \leq 1$,
Lemma~\ref{lemma3} implies
$
\sum^{}_{A \in \Acal}|\delta_T A|\,\eps 
           + \textcolor{two}{2}\,|\Acal[\barra{T}]|\,\eps
\leq 
2\,|\Acal\setminus\Acal_{\{\tstar\}}|\,\eps
$,
as in the previous case.
The addition of this inequality to~(\ref{n-eq3})
produces
\[
\mytextstyle
\sum^{}_{e\in E_T} \yeps(\Lcal'(e)) 
 + \textcolor{two}{2}\,\yeps(\Lcal'[\barra{T}])  
\ \leq \ 
2\,\yeps(\Lcal' \setminus \Lcal'_{\{\tstar\}})\:,
\]
since $\yeps_{L_1 \cup L_2} = 0$
and $\yeps$ differs from $y$ only in~$\Acal$.
Hence, (i6) remains true at the beginning of the next iteration.~\qed

\subparagraph{Proof of (i7).}
Suppose we are at the beginning of the first iteration of phase~II.
Let $L$ be an element of $\Lcal$ such that $L\cap V_T\neq \emptyset$.
Since $V_T=M\in\Lcal^*$, % and $\Lcal$ is laminar
we have $L\subseteq V_T$
and therefore $T[V_T\cap L] = T[L]=F[L]$.
Since $F[L]$ is connected by virtue of (i1)\mar{(i1)},
so is $T[V_T\cap L]$.
This argument shows that $T$ is $\Lcal$-connected.
In particular, $T$ is $M$-connected and therefore $T$ is a tree. 
Hence, (i7) holds at the beginning of the first iteration. 

%%%%%%%%%%%%%%%% also good
% Suppose we are at the beginning of the first iteration of phase~II.
% Let $L$ be an element of $\Lcal$ 
% and let $x$ and $y$ be two vertices in~$V_T\cap L$,
% which is equal to $M\cap L$.
% By virtue of invariant~(i1)\mar{(i1)}.
% there is a path, say $P$, from $x$ to~$y$ in~$F[M]$.
% Since $F$ is forest, $P$ is the unique path form $x$ to $y$ in~$F$
% as well as the unique path from $x$ to $y$ in~$T$.
% By virtue of (i1) applied to~$L$,
% path $P$ neves leaves~$L$.
% This argument shows that (i7) holds at the beginning of the first iteration. 

Now suppose (i7) holds at the beginning of some iteration
where Case~II.1 occurs.
Let $L$ be an element of~$\Lcal$
and let $u$ and $v$ be vertices in $L\cap (V_T\setminus Z)$.
Let $P$ be the unique path from $u$ to $v$ in~$T$.
We may assume that $P$ never leaves~$L$.
Moreover, 
$P$ never enters~$Z$,
given that $|\delta_T Z|=1$.
Hence, $T-Z$ is $L$-connected.
For the same reason, $T-Z$ is a tree.
Hence (i7) holds at the beginning of the next iteration.~\qed

%%%%%%%%%%%%%%% delete soon
% \subparagraph{Proof of (old i8).}
%%% (((old i8) $T$ is not wrapped in~$\Scal$;))
% Invariant (old i8) holds at the beginning of the first iteration because
% $M$ is a maximal element of $\Lcal$
% % $\Lcal$ is laminar,
% and $M\notin \Scal$.
% 
% Now suppose we are at the end of an iteration where Case~II.1 occurs.
% Suppose for a moment that $T-Z$ is wrapped in~$\Scal$.
% Since $\Lcal$ is laminar and $M\notin\Scal$,
% by (i8A)\mar{(i8A)},
%%% (i8A) used to be this: 
%%%     $M\setminus V_T$ can be partitioned into elements of $\Scal$;
% $M$ is the union of elements of $\Scal$.
% But this contradicts~(i5)\mar{(i5)}.
% Hence,
% $T-Z$ is not wrapped in~$\Scal$
% and therefore (old i8) holds at the beginning of the next iteration.~\qed

\color{myred}
\subparagraph{Proof of (i8).}
At the beginning of the first iteration of phase~II,
(i8) holds because $V_T=M$.
Now consider an iteration where Case~II.1 occurs.
We may assume that there is a partition $\Ucal$ of $M\setminus V_T$ 
into elements of~$\Scal$.
If $Z\subseteq V_T$ then $\Ucal\cup\{Z\}$ is a partition of 
$M\setminus (V_T\setminus Z)$
into elements of~$\Scal$.
Otherwise,
% since $\Lcal$ is laminar,
$Z$ includes some of the elements of $\Ucal$ 
and is disjoint from all the others.
\begin{hide}
 ((Note that $V_T$ need not be an element of $\Lcal$.))
\end{hide}
Hence,
$\{Z\} \cup \{U\in \Ucal: U\cap Z=\emptyset\}$
is a partition of 
$M\setminus (V_T\setminus Z)$
into elements of~$\Scal$.
This shows that (i8) holds at the beginning of the next iteration.~\qed
\color{black}

%%%%%%%%%%%%%%% delete soon
% \subparagraph{Proof of (i10).}
% %% Old (i10): $\pi(\barra{T}) = y(\Lcal[\barra{T}])$;
% Invariant (i10) holds at the beginning of the first iteration because
% $M$ is the only element of $\Lcal^*\setminus\Scal$ and therefore 
% \[
% \mytextstyle
% \pi(\barra{T})
%  =  \pi(\barra{M})
%  =  \sum^{}_{S\in \Lcal^*\cap \Scal} \pi(S)
%  =  \sum^{}_{S\in \Lcal^*\cap \Scal} y(\Lcal[S])
%  =  y(\Lcal[\barra{M}])
%  =  y(\Lcal[\barra{T}])\:,
% \]
% where the third equality follows from~(i4).\mar{(i4)}
% Now suppose (i10) is true at the beginning of some iteration.
% At the end of Case~II.1,
% due to~(i11)\mar{(i11)},
% we have $Z\subseteq V_T$.
% By~(i4)\mar{(i4)},
% $\pi(Z) = y(\Lcal[Z])$.
% Since $\barra{V_{T-Z}} = \barra{V_T} \cup Z$,
% property (i10) is still true when the next iteration begins.~\qed

%%%%%%%%%%%%%%%%% delete soon
% \subparagraph{Proof of (i11).}
% %% Old (i11): for each $S$ in $\Scal$, 
% %%            if $|\delta_T S|=1$ then $S\subseteq V_T$.
% Invariant~(i11) holds at the beginning of the first iteration because 
% $V_T=M\in \Lcal$ and $\Lcal$ is laminar.
% 
% Now suppose we are at the end of an iteration where Case~II.1 occurs.
% Let $S$ be any element of $\Scal$ such that
% $|\delta_{T-Z} S|=1$.
% We may assume $S\subseteq V_T$.
% Due to the maximality of~$Z$,
% the set $S$ is disjoint from~$Z$.
% we have 
% and $S\not\subseteq Z$
% is disjoint from~$Z$.~\qed

\normalem
\bibliographystyle{plain}

%\bibliography{/home/mac/cris/aprox/aprox}   % .bib
% \bibliography{../aprox}   % .bib

\end{document}